\documentclass[11pt,letterpaper]{article}

\usepackage[T1]{fontenc}
\usepackage{lmodern}
\usepackage{amsmath,amssymb,amsthm}
\usepackage{thm-restate}
\usepackage[margin=1in]{geometry}
\usepackage{xcolor}
\usepackage[numbers,sort&compress]{natbib}
\usepackage[colorlinks=true,allcolors=blue!45!black]{hyperref}
\usepackage[capitalise]{cleveref}

\newtheorem{theorem}{Theorem}
\newtheorem{lemma}[theorem]{Lemma}
\crefname{lemma}{Lemma}{Lemmas}
\Crefname{lemma}{Lemma}{Lemmas}

\crefname{proposition}{Proposition}{Propositions}
\Crefname{proposition}{Proposition}{Propositions}
\newtheorem*{maintheorem}{Theorem~\ref{thm:main}}

\newcommand{\F}{\mathbb F_2}
\newcommand{\Ostar}{O^*}
\newcommand{\vct}[1]{\boldsymbol{#1}}

\DeclareMathOperator{\poly}{poly}

\title{Breaking the $2^n$ barrier for directed hamiltonicity}
\author{Tomohiro Koana\thanks{Graduate School of Information Science and Technology, The University of Tokyo, Japan.\newline Email: \texttt{tomohiro.koana@gmail.com}.}
\and Soh Kumabe\thanks{CyberAgent, Inc., Tokyo, Japan. Email: \texttt{kumabe\_soh@cyberagent.co.jp}.}}
\date{}

\begin{document}
\begin{titlepage}
\pagenumbering{roman}
\maketitle
\thispagestyle{empty}

\begin{abstract}
\normalsize
We give a randomized algorithm for Directed Hamiltonian Cycle
on $n$-vertex directed graphs that runs in
$\Ostar(1.9133^n)$ time.
For general directed graphs, this is the first improvement in the exponential
base over the classical $O(2^n n^2)$-time algorithms of Bellman and
Held--Karp (1962).

To obtain this improvement, we reduce Hamiltonicity to weighted parity
counting using the isolation lemma of Mulmuley, Vazirani, and Vazirani
(STOC 1987).  If a Hamiltonian path from $s$ to $t$ exists, assigning
random polynomially bounded integer weights to the arcs gives a unique minimum-weight
Hamiltonian path from $s$ to $t$ with constant probability.
In particular, the path
count at its total weight is odd.

To compute these parities, we give a $(2-2^{-d})^n n^{O(1)}$-time
algorithm for instances in which the arcs entering each vertex have at
most $d$ distinct polynomially bounded integer weights.
The algorithm applies the random linearization idea of Arvind and
Guruswami (IPEC 2021) to the Laplacian determinant construction of
Bj\"orklund, Kaski, and Koutis (ICALP 2017).
To keep $d$ small, we randomly delete and duplicate arcs and partition
the incoming copies at each vertex into $d$ groups, assigning a common
weight to all arcs in each group.  We prove that, if the input graph has a
Hamiltonian path from $s$ to $t$, then with positive, albeit exponentially
small, probability one can select one group at each vertex so that their
union contains an odd number of such paths.

\end{abstract}
\end{titlepage}
\pagenumbering{arabic}

\section{Introduction}\label{sec:introduction}

Hamiltonian Cycle asks whether a graph on $n$ vertices contains a cycle that
visits every vertex exactly once.  The classical dynamic programs of
Bellman~\cite{Bellman1962} and Held and Karp~\cite{HeldKarp1962} solve the
problem in $O(2^n n^2)$ time.

For undirected graphs, a celebrated result of Bj\"orklund~\cite{Bjorklund2014}
broke the $2^n$ barrier by giving a randomized $O(1.657^n)$-time algorithm
based on determinant sums.  This work earned Bj\"orklund the 2016
EATCS--IPEC Nerode Prize.  For directed graphs, however, whether Hamiltonian
Cycle can be solved in
$O((2-\varepsilon)^n)$ time for some constant $\varepsilon>0$ has
remained open since 1962
(see, e.g.,~\cite{BjorklundKaski2024,CyganKratschNederlof2018,HusfeldtEtAl2013,KowalikEtAl2020,KowalikMajewski2020,Woeginger2003}).

We break this barrier by giving a randomized algorithm that runs in
$O(1.9133^n)$ time.

\newcommand{\mainTheoremStatement}{%
Directed Hamiltonian Cycle on an $n$-vertex directed graph has a one-sided randomized
algorithm with
running time $O(1.9133^n)$.}
\begin{theorem}\label{thm:main}
\mainTheoremStatement
\end{theorem}

\paragraph{Our approach.}
We reduce Directed Hamiltonian Cycle to weighted parity counting for
Hamiltonian paths between prescribed endpoints $s$ and $t$.
The isolation lemma of Mulmuley, Vazirani, and
Vazirani~\cite{MulmuleyVaziraniVazirani1987} implies that, if such a path
exists, assigning independent random polynomially bounded integer weights
to the arcs gives a unique minimum-weight path with constant probability.
The path count at its total weight is then odd.
To compute these parities, we give a $(2-2^{-d})^n n^{O(1)}$-time algorithm
when the arcs entering each vertex have at most $d$ distinct polynomially
bounded integer weights.  The algorithm applies the random linearization
idea of Arvind and Guruswami~\cite[Remark~16]{AG2021} to the Laplacian
determinant construction of Bj\"orklund, Kaski, and Koutis~\cite{BKK2017}.

A direct application of the isolation lemma may lead to $d=\Theta(n)$.
To keep $d$ small, we randomly
delete and duplicate arcs and partition the incoming copies at each vertex
into $d$ groups.  We assign independent random polynomially bounded
integer weights to the groups, giving all arcs in a group the same weight.
We prove that, if a Hamiltonian path from $s$ to $t$
exists, then with positive, albeit exponentially small, probability there
is a selection of one group at each vertex (other than $s$)
such that the union of the selected groups contains an odd number of such paths.
Applying the isolation lemma to the selections with odd path counts gives a
unique selection of minimum total weight with constant probability.
All other selections at that weight have even path counts and cancel modulo
two, leaving an odd total path count.
For $d=3$, parity counting takes $\Ostar((15/8)^n)$ time per trial,\footnote{The $\Ostar$
notation suppresses factors polynomial in $n$.} and
$O((50/49)^n)$ independent trials suffice for constant success probability.
Together they give the running time $O(1.9133^n)$.
We describe the algorithm in more detail in \Cref{sec:overview}.

\paragraph{Related work.}
The dynamic programs of Bellman~\cite{Bellman1962} and Held and
Karp~\cite{HeldKarp1962} solve both directed and undirected Hamiltonicity in
$\Ostar(2^n)$ time. Karp~\cite{Karp1982} and Bax~\cite{Bax1993} later gave
$\Ostar(2^n)$-time inclusion--exclusion formulations.  For undirected graphs,
Bj\"orklund~\cite{Bjorklund2014} was the first to reduce the base of the
exponential running time below $2$, with a randomized $O(1.657^n)$-time
algorithm.  The
$O(1.657^n)$ running time remains the fastest known bound for general
undirected graphs.  For several restricted classes of directed graphs,
Directed Hamiltonian Cycle can also be solved in
$O((2-\varepsilon)^n)$ time for a constant $\varepsilon>0$.
Bj\"orklund,
Husfeldt, Kaski, and
Koivisto~\cite{BjorklundHusfeldtKaskiKoivisto2012} gave an
$O((2-\varepsilon_\Delta)^n)$-time algorithm for the asymmetric Traveling
Salesman Problem, which is the weighted optimization version of Directed
Hamiltonian Cycle, on instances whose underlying undirected graph has maximum
degree $\Delta$,
where $\varepsilon_\Delta>0$ depends only on $\Delta$.  Cygan, Kratsch, and
Nederlof~\cite{CyganKratschNederlof2018} gave a randomized
$O(1.888^n)$-time algorithm for Hamiltonian Cycle on bipartite directed
graphs.  Bj\"orklund, Kaski, and Koutis~\cite{BKK2017} gave a randomized
$\Ostar(3^{n/2})$-time algorithm for Hamiltonian Cycle on
bipartite directed graphs.  Kowalik and
Majewski~\cite{KowalikMajewski2020} gave a randomized
$\Ostar((2-2^{-a})^n)$-time algorithm for Directed
Hamiltonian Cycle on graphs of average outdegree $a$.
Bj\"orklund~\cite{Bjorklund2021} gave a randomized
$2^{n-\Omega(n/a)}$-time decision algorithm for Directed
Hamiltonian Cycle on graphs of average outdegree $a$.
For general directed graphs, \mbox{Bj\"orklund and Kaski}~\cite{BjorklundKaski2024}
and Pratt~\cite{Pratt2024}
showed that Strassen's asymptotic rank
conjecture implies a randomized $O((2-\varepsilon)^n)$-time algorithm
for some constant $\varepsilon>0$.

Another line of work concerns counting variants.  Bj\"orklund and
Husfeldt~\cite{BjorklundHusfeldt2013} gave a deterministic $O(1.619^n)$-time
algorithm for computing the parity of the number of directed Hamiltonian
cycles.  Bj\"orklund, Kaski, and Koutis~\cite{BKK2017} extended modular
counting to prime powers in $O((2-\varepsilon)^n)$ time for some constant
$\varepsilon>0$, using the Laplacian determinant sum that underlies our
evaluation algorithm.  For exact counting,
Bj\"orklund~\cite{Bjorklund2016} gave a
$2^{n-\Omega(\sqrt{n/\log n})}$-time algorithm.  Bj\"orklund, Kaski, and
Williams~\cite{BjorklundKaskiWilliams2019} reduced the running time to
$2^{n-\Omega(\sqrt{n/\log\log n})}$, and Li~\cite{Li2026} reduced the running
time further to $2^{n-\Omega(\sqrt n)}$.  Under Strassen's asymptotic rank
conjecture, the results of Bj\"orklund, Kaski, Koana, and
Nederlof~\cite{BjorklundKaskiKoanaNederlof2026} imply that directed
Hamiltonian cycles can be counted exactly in $O((2-\varepsilon)^n)$ time
for some constant $\varepsilon>0$.

The $k$-Path problem, which asks whether a graph contains a simple
path on $k$ vertices, is a parameterized generalization of Hamiltonian Path.
Monien~\cite{Monien1985} gave an $\Ostar(k!)$-time algorithm, and Alon, Yuster, and
Zwick~\cite{AlonYusterZwick1995} introduced color-coding and obtained the first
single-exponential FPT algorithm.  Subsequent improvements used divide-and-color
\cite{ChenEtAl2009,KneisEtAl2006}, algebraic sieving and group algebras
\cite{Koutis2008,KoutisWilliams2016,Williams2009}, mixed color-coding and
representative families~\cite{FominEtAl2016,Nederlof2025,Tsur2019,Zehavi2015},
and extensor coding~\cite{BrandDellHusfeldt2018}.  For directed $k$-Path,
Williams~\cite{Williams2009} and Koutis and Williams~\cite{KoutisWilliams2016}
achieved a randomized running time of $\Ostar(2^k)$, which is the fastest bound
currently known.  For undirected $k$-Path, Bj\"orklund, Husfeldt, Kaski, and
Koivisto~\cite{BjorklundHusfeldtKaskiKoivisto2017} broke the $\Ostar(2^k)$
barrier and obtained a randomized running time of $\Ostar(1.657^k)$, matching
the $O(1.657^n)$ bound for undirected Hamiltonian Path.

\paragraph{Organization.}
The rest of this paper is organized as follows.
\Cref{sec:overview} gives an outline of the algorithm.
\Cref{sec:projection} proves that, with sufficient probability, selecting
one group of incoming arcs at each vertex gives an odd number of
Hamiltonian paths.
\Cref{sec:algorithm} shows how to compute the parity of the path count at
each total weight.

\section{The algorithm}\label{sec:overview}

Our starting point is the $O(1.619^n)$-time algorithm of
Bj\"orklund and Husfeldt~\cite{BjorklundHusfeldt2013} for counting directed
Hamiltonian cycles modulo two.
However, this does not directly yield an $O((2-\varepsilon)^n)$-time
algorithm for the decision problem: we must still distinguish a graph
with no Hamiltonian cycles from one with a positive even number of them.
A standard approach is to assign independent random weights to the arcs and
count cycles modulo two at each total weight,
using the \emph{isolation lemma} of Mulmuley, Vazirani, and Vazirani \cite{MulmuleyVaziraniVazirani1987}.
If a Hamiltonian cycle exists, random polynomially bounded integer
arc weights \emph{isolate} one with constant probability: there is a
unique minimum-weight Hamiltonian cycle.
In particular, the cycle count at its
total weight becomes odd.

\begin{lemma}[Isolation lemma~{\cite{MulmuleyVaziraniVazirani1987}}]
\label{lem:isolation}
Let $U$ be a set of $N$ elements and let $\mathcal F\subseteq2^U$ be
nonempty.  Assign each $u\in U$ an integer weight $\omega(u)$, independently
and uniformly from $\{1,\ldots,2N\}$.  With probability at least $1/2$,
there is a unique set $I\in\mathcal F$ minimizing
$\sum_{u\in I}\omega(u)$.
\end{lemma}

The isolation lemma was originally introduced to reduce matching to
determinant computation~\cite{MulmuleyVaziraniVazirani1987}.
It also implies, for example, that finding the unique satisfying assignment of a SAT
instance is as hard as SAT under randomized reductions~\cite{ValiantVazirani1986}.
We use the isolation lemma to exploit faster algorithms for parity counting.
This way of exploiting faster parity counting is also used by
Cygan et al.~\cite{CyganEtAl2011} in their Cut\&Count framework to solve
connectivity problems in time single-exponential in treewidth.

This strategy, however, does not directly give a faster algorithm for Hamiltonicity
since the parity-counting algorithm of Bj\"orklund and Husfeldt~\cite{BjorklundHusfeldt2013} only handles the unweighted case.

Bj\"orklund, Dell, and Husfeldt~\cite{BjorklundDellHusfeldt2015} gave
an alternative reduction to parity counting based on \emph{random restrictions},
where each arc is deleted independently with probability $p$.
For example, when $p=1/2$, the remaining graph contains an odd number of Hamiltonian cycles
with probability at least $2^{-n}$ whenever the input graph is Hamiltonian.
For a suitable choice of $p$, this reduction and the parity-counting
algorithm of Bj\"orklund and Husfeldt~\cite{BjorklundHusfeldt2013} yield a randomized
$O(1.9999^n)$-time algorithm for directed $n$-vertex graphs with at most
$1.0385^n$ Hamiltonian cycles.

Our algorithm extends this random-restriction approach and combines it
with the isolation lemma~\cite{MulmuleyVaziraniVazirani1987}.
Not only do we delete arcs at random, but we also add parallel arcs
that receive different weights.

We now give the formal construction.
Let $G=(V,A)$ be a simple loopless directed graph on $n$ vertices,
with distinct endpoints $s,t$ and no arcs entering $s$ or leaving $t$.
Fix an integer $d \ge 2$ (any such $d$ gives an $O((2-\varepsilon)^n)$-time
algorithm, and we will use $d=3$, which gives the best running time) and define
$[d]:=\{1,\ldots,d\}$ and $\delta:=1/(1+(2^d-1)^2)$.
For each vertex $v\ne s$, initialize $d$ empty groups
$\widehat A_{v,1},\ldots,\widehat A_{v,d}$ of arcs entering $v$.
(All arcs in each group $\widehat A_{v,\ell}$ will receive the same weight.)
Independently delete each original arc with probability $\delta$.
For each retained arc $e=uv$, independently choose a nonempty subset of
$[d]$ uniformly and place one copy $e_\ell$ of the arc $uv$ in
$\widehat A_{v,\ell}$ for each chosen index $\ell$.
In other words, for each $e=uv\in A$ and $X\subseteq[d]$,
\[
 \Pr\!\left[\{\ell\in[d]:e_\ell\in\widehat A_{v,\ell}\}=X\right]=
 \begin{cases}
  \delta, & \text{ if } X=\varnothing,\\[2pt]
  \dfrac{1-\delta}{2^d-1}, & \text{ if } X\ne\varnothing.
 \end{cases}
\]
Let $\widehat G=(V,\widehat A)$ be the multigraph consisting of all these
copies.

We show that, if $G$ has a Hamiltonian path from $s$ to $t$, some choice
of groups yields an odd path count, albeit with an exponentially small
guaranteed success probability.  We prove the following proposition
in \Cref{sec:projection}.

\begin{restatable}{proposition}{projectionproposition}\label{prop:projection}
If $G$ has a Hamiltonian path from $s$ to $t$, then, with probability at
least $(1-\delta)^{n-1}$, there exists a tuple of indices
$(\ell_v)_{v\ne s}\in[d]^{n-1}$ such that
$(V,\bigcup_{v\ne s}\widehat A_{v,\ell_v})$ has an odd number of
Hamiltonian paths from $s$ to $t$.
\end{restatable}

Choose a random weight $w_{v,\ell}$
from $\{1,\ldots,2d(n-1)\}$ for each group $\widehat A_{v,\ell}$, and assign
the weight $w_{v,\ell}$ to every arc in $\widehat A_{v,\ell}$.
By \Cref{prop:projection,lem:isolation}, if $G$ has a Hamiltonian path
from $s$ to $t$, then, with probability at least $\frac12(1-\delta)^{n-1}$,
there is a unique tuple of indices $(\ell_v)_{v\ne s}\in[d]^{n-1}$
that minimizes $\sum_{v\ne s}w_{v,\ell_v}$ among all tuples for which
$(V,\bigcup_{v\ne s}\widehat A_{v,\ell_v})$ has an odd number of
Hamiltonian paths from $s$ to $t$.
Let $B:=\sum_{v\ne s}w_{v,\ell_v}$.
For every other tuple with total weight $B$, the corresponding subgraph
has an even number of Hamiltonian paths from $s$ to $t$, which are cancelled out modulo two.
So $\widehat G$ has an odd number of
Hamiltonian paths from $s$ to $t$ of total weight $B$.

The parity-counting algorithm, which only handles unweighted instances, still does not apply. We therefore extend it to the weighted setting,
exploiting the fact that the arcs entering each vertex have at most $d$
distinct weights:

\begin{restatable}{proposition}{paritycountinglemma}\label{lem:parity-counting}
Let $d\ge1$ be an integer, and let $H$ be an $n$-vertex directed multigraph
with polynomially many arcs and integer arc weights in $\{0,\ldots,W\}$.
Suppose that the arcs entering each vertex have at most $d$ distinct weights.
Given distinct vertices $s,t\in V(H)$ and an integer $B$, a randomized
algorithm correctly computes the parity of the number of Hamiltonian paths
from $s$ to $t$ of total weight $B$ with probability at least $5/6$.
Its running time is $(2-2^{-d})^n\poly(n,W)$.
\end{restatable}

For $d=1$, \Cref{lem:parity-counting} gives an $\Ostar(1.5^n)$-time algorithm
for ordinary parity counting.
The algorithm of \Cref{lem:parity-counting} is randomized and uses exponential space,
whereas the aforementioned algorithm of Bj\"orklund and Husfeldt~\cite{BjorklundHusfeldt2013}
is deterministic and uses polynomial space.
We prove \Cref{lem:parity-counting} by
applying the random linearization idea of Arvind and
Guruswami~\cite[Remark~16]{AG2021} to the Laplacian determinant construction
of Bj\"orklund, Kaski, and Koutis~\cite{BKK2017}.

With \Cref{prop:projection,lem:parity-counting} in hand, we are ready
to prove our main theorem.

\begin{maintheorem}
\mainTheoremStatement
\end{maintheorem}

\begin{proof}[Proof of \Cref{thm:main}]
We try each possible closing arc $ts$ and test for a Hamiltonian path from
$s$ to $t$.
We first obtain a two-sided randomized algorithm.
Let $d:=3$ and $N:=d(n-1)$, and sample $\widehat G$ as above.
For every $v\ne s$ and $\ell\in[d]$, choose
$\omega(v,\ell)$ independently and uniformly from $\{1,\ldots,2N\}$,
independently of $\widehat G$.  Give every arc in $\widehat A_{v,\ell}$
weight $\omega(v,\ell)$.  Then the arcs entering any vertex have at most
$d$ distinct weights.  For $0\le B\le2N(n-1)$, let $c_B$ be the parity of
the number of Hamiltonian paths from $s$ to $t$ of total weight $B$.

If a Hamiltonian path from $s$ to $t$ exists, \Cref{prop:projection,lem:isolation}
give a unique minimum-weight group choice among those with odd path counts
with probability at least $\frac12(49/50)^{n-1}$.
At its weight $B$, even counts cancel modulo two, so $c_B=1$.
Thus, after $O((50/49)^n)$ independent trials, some trial has $c_B=1$ for some
weight $B$ with error probability at most $1/6$.

We now apply \Cref{lem:parity-counting} to compute each parity.
We repeat each computation independently $O(n)$ times and take the
majority.  By a union bound, the probability that any computed parity
is incorrect is at most $1/6$.
We answer yes if any computed parity is $1$.
The two error bounds give a two-sided randomized algorithm with error probability at most $1/3$.

The running time is
$\Ostar((50/49)^n\cdot(15/8)^n)=\Ostar((375/196)^n)=O(1.9133^n)$.

To make the error one-sided, consider each arc $e$ of the input graph in turn.
Run the algorithm on the current graph without $e$, and delete $e$ if the
answer is yes.
By repeating each call $O(\log n)$ times, its error can be reduced to at most $1/(3n^2)$.
There are at most $n^2$ calls, so all answers are correct with probability
at least $2/3$.  In this event, a Hamiltonian input leaves exactly a
Hamiltonian cycle.
We accept only if the remaining graph is a Hamiltonian cycle.  This gives one-sided error with only polynomial overhead.
\end{proof}

We remark that our approach does not yield an $\Ostar((2-\varepsilon)^k)$-time algorithm
for directed $k$-Path: an analogue of \Cref{lem:parity-counting} is unlikely,
since counting $k$-paths from $s$ to $t$ modulo two is
$\oplus$W[1]-hard~\cite{CurticapeanDellHusfeldt2021}.

\section{Obtaining an odd number of Hamiltonian paths}\label{sec:projection}

In this section, we prove \Cref{prop:projection}.

Recall that, for each arc $uv\in A$, we place one copy of $uv$ in
$\widehat A_{v,\ell}$ for each $\ell\in S_{uv}$, where $S_{uv}\subseteq[d]$.
Write $\delta:=1/(1+(2^d-1)^2)$ and $\rho:=(2^d-1)\delta$.
The sets $S_{uv}$ are chosen independently, with
\[
 \Pr[S_{uv}=\varnothing]=\delta,\qquad
 \Pr[S_{uv}=S]=\rho\quad(\varnothing\ne S\subseteq[d]).
\]
Note that $\delta^2+\rho^2=\delta$.

\projectionproposition*

We prove the proposition by expressing the parity of each path count
as a coefficient of a polynomial.
For each $v\ne s$ and $\ell\in[d]$, introduce an indeterminate
$z_{v\ell}$ for the group $\widehat A_{v,\ell}$.
Let $\mathcal P_{s,t}(G)$ be the family of Hamiltonian paths from $s$ to
$t$ in $G$, and define over $\F$
\[
 f(z):=\sum_{P\in\mathcal P_{s,t}(G)}
       \prod_{uv\in P}\left(\sum_{\ell\in S_{uv}}z_{v\ell}\right).
\]
Each term in the expansion of the product for a path $P$ corresponds to
choosing one copy of each arc of $P$.
Every Hamiltonian path enters each vertex $v\ne s$ exactly once.
Thus, for any choice of indices $\ell_v\in[d]$ for $v\ne s$, the
coefficient of $\prod_{v\ne s}z_{v\ell_v}$ is the parity of the number
of Hamiltonian paths from $s$ to $t$ using only the selected groups
$\widehat A_{v,\ell_v}$.
Consequently, $f(z)\ne0$ exactly when there is a tuple of indices
$(\ell_v)_{v\ne s}\in[d]^{n-1}$ for which
$(V,\bigcup_{v\ne s}\widehat A_{v,\ell_v})$ has an odd number of
Hamiltonian paths from $s$ to $t$.
Thus, \Cref{prop:projection} is equivalent to the following lemma.

\begin{lemma}\label{lem:polynomial-nonvanishing}
If $G$ has a Hamiltonian path from $s$ to $t$, then
$\Pr[f(z)\ne0]\ge(1-\delta)^{n-1}$.
\end{lemma}

To prove \Cref{lem:polynomial-nonvanishing}, we use the following
bound on polynomial cancellation.

\begin{samepage}
\begin{lemma}\label{lem:polynomial-cancellation}
Let $\mathbb{K}$ be a field and let $p,d\ge1$ be integers.
Let $z_{i\ell}$, for $i\in[p]$ and $\ell\in[d]$, be
indeterminates in a polynomial ring $R$ over $\mathbb{K}$, and let
$f_1,\ldots,f_p\in R$ be nonzero polynomials.
Then the vector space
\[
 W:=\left\{\vct{c}\in \mathbb{K}^{[p]\times[d]}:
       \sum_{i=1}^p\sum_{\ell=1}^d c_{i\ell}z_{i\ell}f_i=0\right\}
\]
has dimension at most $p-1$ over $\mathbb{K}$.
In particular, if $\mathbb{K}=\F$, then $|W|\le2^{p-1}$.
\end{lemma}
\end{samepage}

We first prove \Cref{lem:polynomial-nonvanishing} and defer the proof of
\Cref{lem:polynomial-cancellation} to the end of this section.

\begin{proof}[Proof of \Cref{lem:polynomial-nonvanishing}]
Assume that $G$ has a Hamiltonian path from $s$ to $t$.
We prove $\Pr[f(z)\ne0]\ge(1-\delta)^{n-1}$ by induction on $n$.
For $n=2$, we have $f(z)=\sum_{\ell\in S_{st}}z_{t\ell}$, which is
nonzero exactly when $S_{st}\ne\varnothing$.
This occurs with probability $1-\delta$.

Suppose $n\ge3$.
For each $v\ne t$ with $sv\in A$, let $f_v(z)$ be the polynomial defined
by the same formula for Hamiltonian paths from $v$ to $t$ in $G-s$.  Then
\[
 f(z)=\sum_{\substack{v:\,sv\in A\\v\ne t}}
       \left(\sum_{\ell\in S_{sv}}z_{v\ell}\right)f_v(z).
\]
Let $p:=|\{v:sv\in A,\ v\ne t,\ f_v(z)\ne0\}|$.
Fix a Hamiltonian path in $G$ and let $sv_*$ be its first arc.
Its suffix is a Hamiltonian path from $v_*$ to $t$ in $G-s$, so
induction gives
$\Pr[f_{v_*}(z)\ne0]\ge(1-\delta)^{n-2}$.
In particular, $\Pr[p>0]\ge(1-\delta)^{n-2}$.

Condition on the sets $S_{uv}$ for all arcs with $u\ne s$, with $p>0$.
This fixes the polynomials $f_v(z)$, and only the sets $S_{sv}$ remain
random.  Then
\[
 f(z)=\sum_{v:\,f_v(z)\ne0}\sum_{\ell=1}^d
       c_{v\ell}\bigl(z_{v\ell}f_v(z)\bigr),
 \qquad\text{where }c_{v\ell}:=\mathbf{1}[\ell\in S_{sv}]\in\F.
\]
Applying \Cref{lem:polynomial-cancellation} over $\F$ to these $p$ nonzero
polynomials shows that $f(z)=0$ for at most $2^{p-1}$ coefficient vectors
$\vct{c}:=(c_{v\ell})_{v:\,f_v(z)\ne0,\,\ell\in[d]}$.

The all-zero coefficient vector $\vct{c}=\vct{0}$, for which $f(z)=0$, corresponds to
$S_{sv}=\varnothing$ for every $v$ with $f_v(z)\ne0$ and occurs with
probability $\delta^p$.
Since $\delta\le\rho$, every nonzero coefficient vector occurs with
probability at most $\rho^p$.
Therefore,
\[
 \Pr[f(z)=0]\le\delta^p+(2^{p-1}-1)\rho^p.
\]
For $p=1$, the right-hand side is $\delta$.
For $p\ge2$, we have $2^{p-1}-1\le3^{p-2}$ and
$\rho\le3/10$ (since $d\ge2$), so
\[
 \delta^p+(2^{p-1}-1)\rho^p
 \le\delta^2 + (9/10)^{p-2} \rho^2
 \le\delta^2+\rho^2=\delta.
\]
Thus, conditioning on $p > 0$, the probability
that $f(z)\ne0$ is at least $1-\delta$, and
\(
 \Pr[f(z)\ne0]\ge(1-\delta)\Pr[p>0]\ge(1-\delta)^{n-1}.
\)
\end{proof}

We conclude this section with the proof of \Cref{lem:polynomial-cancellation}.

\begin{proof}[Proof of \Cref{lem:polynomial-cancellation}]
For $i\in[p]$ and $\ell\in[d]$, let $g_{i\ell}:=z_{i\ell}f_i$, and let
$S$ be the linear span of these $pd$ polynomials over $\mathbb{K}$.
The linear map $T:\mathbb{K}^{[p]\times[d]}\to R$ defined by
$T(\vct{c}):=\sum_{i=1}^p\sum_{\ell=1}^d c_{i\ell}g_{i\ell}$ has image $S$
and kernel $W$.
By rank--nullity, $\dim_{\mathbb{K}} W=pd-\dim_{\mathbb{K}} S$, so it suffices to prove
$\dim_{\mathbb{K}} S\ge p(d-1)+1$.

Fix a lexicographic monomial order, with respect to an arbitrary
ordering of its variables.
It suffices to show that there are at least $p(d-1)+1$
leading monomials among the polynomials $g_{i\ell}$.
For each $i\in[p]$, let $\vct{v}_i\in\mathbb Z_{\ge0}^{[p]\times[d]}$ be the
exponent vector of the leading monomial of $f_i$, restricted to the
variables $z_{j\ell}$ with $(j,\ell)\in[p]\times[d]$.
For $i\in[p]$ and $\ell\in[d]$, let
$\vct{e}_{i\ell}\in\mathbb Q^{[p]\times[d]}$ be the unit vector in coordinate
$(i,\ell)$.
Then, the leading monomial of $g_{i\ell}$ in this lexicographic order has exponent
vector $\vct{m}_{i\ell}:=\vct{v}_i+\vct{e}_{i\ell}$.
We show that these exponent vectors have rank at least $p(d-1)+1$
over $\mathbb Q$, which implies that there are at least $p(d-1)+1$ distinct leading monomials.

Let $U$ be the linear span of the vectors $\vct{m}_{i\ell}$ over $\mathbb Q$.
We claim that the vectors
\[
 \vct{e}_{i\ell}-\vct{e}_{i1}\quad(i\in[p],\ \ell\in\{2,\ldots,d\}),
 \qquad\text{and}\qquad \vct{m}_{11}=\vct{v}_1+\vct{e}_{11}
\]
lie in $U$ and are linearly independent over $\mathbb Q$.
Indeed, $\vct{e}_{i\ell}-\vct{e}_{i1}=\vct{m}_{i\ell}-\vct{m}_{i1}\in U$.
These difference vectors are linearly independent and have coordinate
sum zero, whereas $\vct{m}_{11}$ has positive coordinate sum.
Hence, $\dim_{\mathbb Q}U\ge p(d-1)+1$.

Consequently, at least $p(d-1)+1$ of the vectors $\vct{m}_{i\ell}$ are
distinct, yielding at least as many distinct leading monomials.
The corresponding $g_{i\ell}$ are linearly independent over $\mathbb{K}$:
the largest leading monomial in any nontrivial combination cannot cancel.
Thus, $\dim_{\mathbb{K}} S\ge p(d-1)+1$, and
\[
 \dim_{\mathbb{K}} W=pd-\dim_{\mathbb{K}} S
 \le pd-(p(d-1)+1)=p-1.\qedhere
\]
\end{proof}

\section{The determinant-sum algorithm}\label{sec:algorithm}

We now prove the weighted parity-counting lemma used in \Cref{sec:overview}.

\paritycountinglemma*

% The algorithm computes a univariate polynomial whose coefficient at $y^B$
% is the required parity for total weight $B$.
% In \Cref{sec:determinants}, we express this polynomial as a determinant sum
% using the method of Bj\"orklund, Kaski, and Koutis~\cite[Section~4]{BKK2017}.
% In \Cref{sec:covers}, we apply the random linearization idea of
% Arvind and Guruswami~\cite[Remark~16]{AG2021} to this construction to compute the sum.

Let $H=(V,E)$ be the input multigraph, with one integer
weight $w(e)\in\{0,\ldots,W\}$ assigned to each arc $e\in E$.
We may assume that $H$ has no loops, no arcs entering $s$, and no arcs
leaving $t$, since such arcs occur in no Hamiltonian path from $s$ to $t$.
Since we count modulo two, pairs of parallel arcs with the same weight
can be deleted without changing the parity at any total weight.
Thus, we may assume that parallel arcs have distinct weights.
Let $\mathcal P_{s,t}(H)$ be the family of Hamiltonian paths from $s$ to $t$
in $H$.
Assign each arc the monomial $p(e):=y^{w(e)}\in\F[y]$.
Let $c_B\in\F$ be the parity of the number of such paths of total weight $B$.
Define $h(y)\in\F[y]$ by
\[
 h(y):=\sum_{B=0}^{(n-1)W}c_By^B
      =\sum_{P\in\mathcal P_{s,t}(H)}\prod_{e\in P}p(e).
\]

\subsection{Polynomial reformulation}\label{sec:determinants}

We recall the determinant-sum reformulation of Bj\"orklund, Kaski, and
Koutis~\cite{BKK2017}, which expresses $h(y)$ as a sum of
Laplacian determinants via the directed Matrix--Tree Theorem.

Let $E_{uv}(H)$ be its set of
arcs from $u$ to $v$.
For $p\colon E(H)\to\F[y]$, define the Laplacian $L(H,p)$ by
\[
 L(H,p)_{uv}:=
 \begin{cases}
  \displaystyle\sum_{w\in V(H)}\sum_{e\in E_{wu}(H)}p(e),
       & u=v,\\[4pt]
  \displaystyle-\sum_{e\in E_{uv}(H)}p(e),
       & u\ne v.
 \end{cases}
\]
The minus sign can be omitted since we work over $\F$.
For $s\in V(H)$, let $L_s(H,p)$ be the matrix
obtained by deleting row and column $s$.

An \emph{out-arborescence} of $H$ rooted at $s$ is a spanning tree
with all arcs directed away from $s$.
We use the following weighted form of the directed Matrix--Tree Theorem.

\begin{lemma}[Directed Matrix--Tree Theorem~{\cite[Section~11]{GesselStanley1995}}]
\label{lem:matrix-tree}
For every loopless directed multigraph $H$, map
$p\colon E(H)\to\F[y]$, and vertex $s\in V(H)$,
\[
 \det L_s(H,p)=\sum_T\prod_{e\in T}p(e),
\]
where $T$ ranges over all out-arborescences of $H$ rooted at $s$.
\end{lemma}

We now apply inclusion--exclusion to extract the Hamiltonian paths from
$s$ to $t$. Let $I:=V\setminus\{s,t\}$ be the set of internal vertices.
For $X\subseteq I$, let $H_X$ be obtained from $H$ by deleting all arcs
leaving vertices in $I\setminus X$.
By \Cref{lem:matrix-tree}, $\det L_s(H_X,p)$ is the sum of the weights
of all out-arborescences of $H$ rooted at $s$ in which every vertex in
$I\setminus X$ is a leaf.
Summing over all $X\subseteq I$ counts an arborescence with $k$ leaves in
$I$ exactly $2^k$ times, so its contribution cancels modulo two when $k>0$.
Hence
\[
 h(y)=\sum_{X\subseteq I}\det L_s(H_X,p).
\]
Evaluating the right-hand side directly requires $2^{|I|}=2^{n-2}$
determinants and therefore does not improve the $\Ostar(2^n)$ bound.
The key idea of Bj\"orklund, Kaski, and Koutis~\cite{BKK2017}
is to make many of these determinants vanish by randomizing the diagonal
entries through the addition of arcs leaving $t$.
Since no Hamiltonian path from $s$ to $t$ uses these arcs, the same
inclusion--exclusion argument gives the following identity for the
augmented graphs.

\begin{lemma}\label{lem:random-det-sum}
Let $A$ be any set of additional arcs leaving $t$, and extend $p$
arbitrarily to $A$.
For $X\subseteq I$, let $H_X^A$ be obtained by adding $A$ to $H_X$.
Then
\[
 h(y)=\sum_{X\subseteq I}\det L_s(H_X^A,p).
\]
\end{lemma}

We choose the added arcs as follows.
For each $i\in I$, let $W_i$ be the set of weights of arcs entering $i$,
and independently choose a uniformly random subset $R_i\subseteq W_i$.
For $R:=(R_i)_{i\in I}$ and $X\subseteq I$, form $H_X(R)$ by adding
to $H_X$ one arc $ti$ of weight $b$ for each $i\in I$ and $b\in R_i$.
For each added arc $e$, set $p(e):=y^{w(e)}$.
By \Cref{lem:random-det-sum}, the resulting determinants sum to $h(y)$
for every $R$.
Fix $R$ throughout the computation.

Fix $X\subseteq I$ and $i\in I\setminus X$. Since no arc leaves $i$ in
$H_X(R)$, only the diagonal entry of row $i$ in $L_s(H_X(R),p)$ can be
nonzero. The added arcs contribute $\sum_{b\in R_i}y^b$ to this entry,
so one of the $2^{|W_i|}$ choices of $R_i$ yields a zero row
and hence a zero determinant.

Let $\mathcal X(R)$ be the family of $X\subseteq I$ for which
$L_s(H_X(R),p)$ has no zero row indexed by $I\setminus X$.
All other determinants vanish. Since the choices of $R_i$ are
independent and $|W_i|\le d$,
\[
 \mathbb E_R[|\mathcal X(R)|]
 \le\sum_{X\subseteq I}(1-2^{-d})^{|I\setminus X|}
 =(2-2^{-d})^{n-2}.
\]
Thus, given $\mathcal X(R)$, computing $h(y)$ reduces to evaluating
$|\mathcal X(R)|$ determinants, at most $(2-2^{-d})^{n-2}$ in expectation
over $R$.

Each determinant and its coefficient vector can be computed in time
polynomial in $n$ and $W$.
Every matrix entry has degree at most $W$, so each determinant has degree
at most $(n-1)W$.
We compute each determinant over $\F[y]$ using Berkowitz's division-free
algorithm~\cite{Berkowitz1984}.
It uses polynomially many additions and multiplications, with all
intermediate polynomials having degree $O(nW)$.
Consequently, given $\mathcal X(R)$, we obtain a randomized algorithm
for computing $h(y)$ with expected running time
$(2-2^{-d})^{n-2}\poly(n,W)$, where the expectation is over $R$.

Up to this point, we have followed the determinant-sum approach of
Bj\"orklund, Kaski, and Koutis~\cite{BKK2017}.
In \Cref{sec:covers}, we apply the random linearization idea of Arvind and
Guruswami~\cite[Remark~16]{AG2021} to enumerate $\mathcal X(R)$ in
$(2-2^{-d})^{n-2}\poly(n,W)$ time.

\subsection{Enumeration via random linear systems}\label{sec:covers}

We give a sampling algorithm for $\mathcal X(R)$ and use it repeatedly
to obtain a randomized enumeration algorithm.
We construct a linear system at random, and each solution corresponds
to one sampled set in $\mathcal X(R)$.
To construct the system, we represent a subset $X\subseteq I$ by
variables $x_i:=\mathbf{1}[i\in X]\in\F$, for $i\in I$.

Recall that $X\in\mathcal X(R)$ exactly when, for every $i\in I$,
either $x_i=1$ or the diagonal entry $L_s(H_X(R),p)_{ii}$ is nonzero.
For each $i\in I$, the diagonal entry is
\[
 L_s(H_X(R),p)_{ii}
 =\sum_{e\in E_{si}(H)}p(e)
  +\sum_{u\in I}x_u\sum_{e\in E_{ui}(H)}p(e)
  +\sum_{b\in R_i}y^b.
\]
This entry is nonzero exactly when the coefficient of $y^b$ is nonzero
for some $b\in W_i$.
Coefficients of $y^b$ with $b\notin W_i$ are zero by construction.
For each $i\in I$, let $d_i:=|W_i|\le d$, and let $\vct{s}_i\in\F^{d_i}$ be
the vector of coefficients of $y^b$ in $L_s(H_X(R),p)_{ii}$, indexed by $b\in W_i$.
For fixed $R$, each coordinate of $\vct{s}_i$ is affine in the
variables $(x_i)_{i\in I}$.
Thus, the condition for $i$ is an OR of the $d_i+1$ affine expressions
($x_i$ and the coordinates of $\vct{s}_i$): at least one of them must be $1$.

% For each $i\in I$, define $\vct{s}_i(X,R)\in\F^{d_i}$ by setting its coordinate
% indexed by $b\in W_i$ equal to the coefficient of $y^b$ in the diagonal
% entry of $L_s(H_X(R),p)$ indexed by $i$.
% These coordinates are affine in the variables $(x_i)_{i\in I}$.
% Then $X\in\mathcal X(R)$ exactly when
% $(x_i,\vct{s}_i(X,R))\ne\vct{0}$ for every $i\in I$.
% The condition for $i$ is an OR of $d_i+1$ binary coordinates: at least one must be $1$.

Following the idea of Arvind and Guruswami~\cite{AG2021}, we replace
each OR with a randomly chosen XOR constraint.
For each $i\in I$, we choose a nonempty subset of the $d_i+1$ expressions
uniformly at random and require their XOR to be $1$.
Since these expressions are affine, the resulting constraints form a linear
system in the variables $(x_i)_{i\in I}$.
For any fixed nonzero vector in $\F^{d_i+1}$, exactly $2^{d_i}$ of nonempty subsets
of its coordinates have XOR $1$.
Thus, the constraint for $i$ is satisfied with probability
$2^{d_i}/(2^{d_i+1}-1)>1/2$.

Formally, for each $i\in I$, choose
$\vct{a}_i=(\alpha_i,\vct{\beta}_i)\in\F^{d_i+1}\setminus\{\vct{0}\}$,
where $\alpha_i\in\F$ and $\vct{\beta}_i\in\F^{d_i}$.
Let $\vct{a}:=(\vct{a}_i)_{i\in I}$.
Let $\mathcal X_{\vct{a}}(R)$ be the family of subsets $X\subseteq I$
whose indicators satisfy the linear system
\begin{equation}\label{eq:affine-system}
 \alpha_i x_i+\vct{\beta}_i^{\top}\vct{s}_i=1
 \qquad(i\in I).
\end{equation}
Every solution of \Cref{eq:affine-system} satisfies all original OR constraints:
if an XOR is $1$, then at least one of its selected expressions is $1$.
Hence $\mathcal X_{\vct{a}}(R)\subseteq\mathcal X(R)$.
We can enumerate $\mathcal X_{\vct{a}}(R)$ in $|\mathcal X_{\vct{a}}(R)| \cdot \poly(n, W)$ time via Gaussian elimination.

To prove \Cref{lem:parity-counting}, we establish two properties of the sampling procedure.
\Cref{lem:cover} gives a lower bound on the probability
that a random $\vct{a}$ samples a fixed $X\in\mathcal X(R)$.
For fixed $\vct{a}$, \Cref{lem:visits} shows that the expected number of
solutions is 1 (over the random choice of $R$).

\begin{lemma}\label{lem:cover}
For each $i\in I$, choose $\vct{a}_i$ independently and uniformly from
$\F^{d_i+1}\setminus\{\vct{0}\}$. Then, for every $X\in\mathcal X(R)$,
\[
 \Pr_{\vct{a}}[X\in\mathcal X_{\vct{a}}(R)]
 =\prod_{i\in I}(2-2^{-d_i})^{-1}
 \ge(2-2^{-d})^{-(n-2)}.
\]
\end{lemma}

\begin{proof}
Fix $X\in\mathcal X(R)$.
For each $i\in I$, the vector $(x_i,\vct{s}_i)\in\F^{d_i+1}$ is
nonzero, so the equation for $i$ in \Cref{eq:affine-system} is a nonzero
linear equation in $\vct{a}_i$.
Exactly $2^{d_i}$ vectors in $\F^{d_i+1}$ satisfy it, and the zero vector
does not.
Thus the equation for $i$ is satisfied with probability
$2^{d_i}/(2^{d_i+1}-1)=(2-2^{-d_i})^{-1}$.
\end{proof}

\begin{lemma}\label{lem:visits}
For every $\vct{a}=(\vct{a}_i)_{i\in I}$ with
$\vct{a}_i\in\F^{d_i+1}\setminus\{\vct{0}\}$ for each $i\in I$,
$\mathbb E_R[|\mathcal X_{\vct{a}}(R)|]=1$.
\end{lemma}

\begin{proof}
Fix $\vct{a}$, and let $Z:=\{i\in I:\vct{\beta}_i=\vct{0}\}$.
For $i\in Z$, we have $\alpha_i=1$, so the equation for $i$ forces $x_i=1$.
There are $2^{|I\setminus Z|}$ subsets $X\subseteq I$ containing $Z$.
Fix any such $X$.
For each $i\notin Z$, the arcs $ti$ contribute the vector
$(\mathbf{1}[b\in R_i])_{b\in W_i}$ to $\vct{s}_i$, which makes $\vct{s}_i$ uniformly random in $\F^{d_i}$.
Since $\vct{\beta}_i\ne\vct{0}$, $\vct{\beta}_i^{\top}\vct{s}_i$ is uniformly random in $\F$, so the equation for $i$ holds with probability $1/2$.
The choices of $R_i$ are independent over $i$, so $X$ satisfies all
equations with probability $2^{-|I\setminus Z|}$.
Summing over all $X\supseteq Z$ gives
$\mathbb E_R[|\mathcal X_{\vct{a}}(R)|]
 =2^{|I\setminus Z|}2^{-|I\setminus Z|}=1$.
\end{proof}

We are now ready to prove \Cref{lem:parity-counting}.

\begin{proof}[Proof of \Cref{lem:parity-counting}]
We assume that $n\ge4$.
Set $M:=\lceil n(2-2^{-d})^{n-2}\rceil$.
For each $i\in I$, choose $R_i$ uniformly from $2^{W_i}$ and
$\vct{a}_i^{(1)},\ldots,\vct{a}_i^{(M)}$ uniformly from
$\F^{d_i+1}\setminus\{\vct{0}\}$.
Enumerate the $M$ families $\mathcal X_{\vct{a}^{(k)}}(R)$,
counting sets with multiplicity. We immediately return 0 if this count exceeds $12M$.
Without this cutoff, \Cref{lem:visits} gives an expected total of $M$ reported sets.
Thus, Markov's inequality bounds the probability of returning early by $1/12$.
The cutoff ensures that at most $M$ linear systems and $O(M)$ reported sets are
processed, so enumeration takes $M\poly(n,W)=(2-2^{-d})^n\poly(n,W)$ time.

For each fixed $R$ and $X\in\mathcal X(R)$, \Cref{lem:cover} and independence
bound the probability that none of the $M$ families contains $X$ by
\[
  (1-(2-2^{-d})^{-(n-2)})^M
  \le \exp(-M(2-2^{-d})^{-(n-2)}) \le e^{-n}.
\]
A union bound then bounds the probability that these families fail to cover
$\mathcal X(R)$ by $2^{n-2}e^{-n}<1/12$.

If the cutoff is not exceeded, let $\mathcal S$ be the family of distinct sets enumerated
and return the coefficients of
$\widetilde h(y):=\sum_{X\in\mathcal S}\det L_s(H_X(R),p)$.
By a union bound, with probability at least $5/6$, the cutoff is not exceeded and
$\mathcal S=\mathcal X(R)$.
On this event, \Cref{lem:random-det-sum} gives $\widetilde h=h$.
\end{proof}

\section*{Acknowledgements}
Tomohiro Koana was supported in part by JST CREST Grant Number JPMJCR24Q2,
Japan, and JST ERATO Grant Number JPMJER2301, Japan.

The authors thank Andreas Bj\"orklund, Carla Groenland, and Jesper Nederlof
for their feedback and discussions, and in particular for suggesting an
improved presentation of the proof of \Cref{prop:projection}.

\section*{Declaration of generative AI use}
ChatGPT 6 Astra generated the proof of \Cref{thm:main} and was also used to
draft the manuscript.
The authors provided the proposition statements, which give a combinatorial
interpretation of the originally proposed solution.
The authors verified and revised the proof and the manuscript and take full
responsibility for the paper.

\bibliographystyle{alphaurl}
\begingroup
\raggedright
\bibliography{references}
\endgroup

\end{document}